\documentclass[10pt, journal]{IEEEtran}
\usepackage{amsmath}
\usepackage{amsthm}
\usepackage{cite}
\usepackage{graphicx}
\usepackage{xspace}

\newtheorem{lemma}{Lemma}
\newtheorem{theorem}{Theorem}

\newcommand{\antilope}{\textsc{antilope}\xspace}
\newcommand{\unit}[1]{\ensuremath{\, \mathrm{#1}}}

\begin{document}

\title{Antilope -- A Lagrangian Relaxation Approach to the \textit{de novo} Peptide Sequencing Problem}
\author{Sandro Andreotti, Gunnar W.\ Klau\IEEEauthorrefmark{1}, Knut Reinert\IEEEauthorrefmark{1}

\IEEEcompsocitemizethanks
{
\IEEEcompsocthanksitem S.\ Andreotti and K.\ Reinert are with the Department of Computer Science, Freie Universit{\"a}t Berlin, Germany and the 
International Max Planck Research School for Computational Biology and Scientific Computing, Berlin, Germany\protect\\
E-mail: andreott@inf.fu-berlin.de
\IEEEcompsocthanksitem G.\ W.\ Klau is with the CWI, Life Sciences Group, Amsterdam, the Netherlands, and the Netherlands Institute for Systems Biology
\IEEEcompsocthanksitem \IEEEauthorrefmark{1}shared last authors}
\thanks{}}




\IEEEcompsoctitleabstractindextext{%
\begin{abstract}
  Peptide sequencing from mass spectrometry data is a key step in
  proteome research. Especially \textit{de novo} sequencing, the
  identification of a peptide from its spectrum alone, is still a
  challenge even for state-of-the-art algorithmic approaches. In this
  paper we present \antilope, a new fast and flexible approach based
  on mathematical programming. It builds on the spectrum graph model
  and works with a variety of scoring schemes. \antilope combines
  Lagrangian relaxation for solving an integer linear programming
  formulation with an adaptation of Yen's $k$ shortest paths
  algorithm. It shows a significant improvement in running time
  compared to mixed integer optimization and performs at the same
  speed like other state-of-the-art tools. We also implemented a
  generic probabilistic scoring scheme that can be trained
  automatically for a dataset of annotated spectra and is independent
  of the mass spectrometer type. Evaluations on benchmark data show
  that \antilope is competitive to the popular state-of-the-art
  programs PepNovo and NovoHMM both in terms of run time and
  accuracy. Furthermore, it offers increased flexibility in the number
  of considered ion types. \antilope will be freely available as part of
  the open source proteomics library OpenMS.
\end{abstract}
}

\maketitle

\IEEEdisplaynotcompsoctitleabstractindextext

\section{Introduction}
Mass spectrometry-based high throughput identification of peptides and
proteins is a key step in most proteomics research experiments. It
requires fast algorithmic solutions with good identification
capabilities.  Depending on
the initial situation of the experiment, two general
strategies exist: database-assisted and \textit{de novo}
identification.  If a database for the studied proteins exists the
first method is usually preferred over \textit{de novo} sequencing.  The
crucial step in database search algorithms like
INSPECT\cite{Tanner_Inspect_2005}, SEQUEST\cite{Sequest_1994}, Mascot\cite{Mascot_99} and OMSSA\cite{OMSSA} is to
filter the database based on different methods.  INSPECT generates
peptide sequence tags (PST) and keeps only those candidate peptides
containing the tag as a subsequence. SEQUEST uses the parent mass as
filter criterion.
After filtering, the query spectrum is scored against the remaining candidates and a ranking of possible identifications is produced.
In addition to the quality of the spectrum, database search methods clearly depend on the correctness and completeness of the database
and hence on the availability of a suitable set of peptides or transcripts for the studied organism. Even if this is the case, factors like alternative splice variants and mutations can lead to missing identifications. 

In such situations \textit{de novo} sequencing algorithms provide an alternative as they infer the sequence from the spectrum itself without any information collected in databases.
In recent years, many algorithms and software packages were published, with the most popular being PEAKS\cite{PEAKS_03}, PepNovo\cite{Pepnovo_05}, NovoHMM\cite{NovoHMM_05}, Lutefisk\cite{Lutefisk_97}, Sherenga\cite{dacvp-denovo:1999}, EigenMS\cite{Bern_EigenMS_2006}, and PILOT\cite{Floudas_DiMaggio_07}.
Most of them use the graph-theoretical approach introduced by Bartels\cite{Bartels:1990} and construct a so-called N-C spectrum graph which is then used to search for the correct sequence. See Fig.~\ref{fig:ncgraph}.

\begin{figure*}[hbt]
\centering
\includegraphics[width=1.0\linewidth]{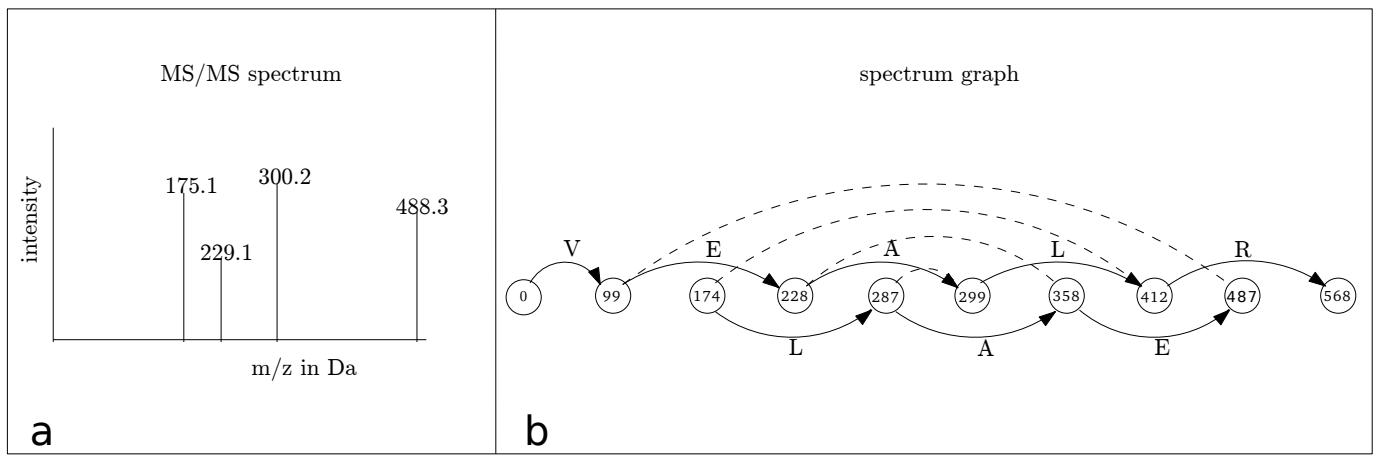}
\caption{Spectrum graph generation.
(a) Simplified tandem mass spectrum of the peptide VEALR\@. Rounded m/z values in Da are presented on top of each peak.
(b) The corresponding spectrum graph with two nodes being generated for each peak.
One under the assumption of being a b-ion, the other under the assumption of being a y-ion.
It it obvious that the path starting at node $s$ with mass 0 and ending at node $t$ with mass 568 encodes the correct peptide sequence.
The undirected edges connecting complementary nodes are drawn as dashed lines.}
\label{fig:ncgraph}
\end{figure*}

Using this formulation, the \textit{de novo} peptide sequencing problem can be formulated as the search for the longest antisymmetric path, an NP-complete problem\cite{Gabow76}. 
PepNovo and Lutefisk solve a special case of this problem by restricting the construction of the spectrum graph, which enables them to apply a dynamic programming algorithm proposed by Chen et al.\cite{Chen_2001, Chen_subopt_03}. The restrictions limit the possible interpretations of each peak to at most one N-terminal (usually b-ion) and one C-terminal (usually y-ion) ion type. Liu and Cai\cite{LiuCaiTreeDecomp} use tree-decomposition to solve the restricted problem. 
Bafna and Edwards\cite{Bafna_Edwards_03}  propose a variant of the dynamic programming approach that also allows for more interpretations leading to a polynomial algorithm of a higher degree.
Their algorithm is still limited to so-called \textit{simple} ion types, excluding doubly and triply charged ions that can also aid the identification process.
PILOT\cite{Floudas_DiMaggio_07} overcomes all these restrictions using an integer linear programming (ILP) formulation for the longest antisymmetric path problem that is flexible and extensible on the cost of efficiency.
This allows for more interpretations of each peak which can lead to improved identification in situations where the prominent b- and y-ions are missing.
Furthermore the ILP formulation can be easily extended in several ways by simply adding or modifying constraints to further restrict or modify the set of possible solutions. The approach also allows for global reasoning such as limiting the number of a certain amino acid type for each prediction.

The main contribution of this work is an improvement of this approach by an extension that retains most of the flexibility and leads to a remarkable improvement in running time.
Instead of focusing on computing one antisymmetric path we propose a novel algorithm to find the $k$ best antisymmetric paths.  We achieve this by applying the Lagrangian relaxation technique to the problem and solving the subproblems with an elegant variant of Yen's $k$ shortest paths algorithm. Lagrangian relaxation was already successfully applied to biological problems such as sequence alignment\cite{DBLP:journals/jco/AlthausC08}, protein\cite{Caprara2004a} and RNA\cite{Bauer2007} structural alignment or protein threading\cite{Threading_Andonov:04}. 

An additional contribution of this paper is a generic probabilistic scoring scheme that can be trained automatically for a dataset of annotated spectra and is independent of the mass spectrometer type. The performance of both, \textit{de novo} and database search approaches, depends on a good scoring function to model prediction quality.
Currently used scoring functions range from rather simple peak intensity-based scoring to statistical models including Bayesian networks.
The latter show a better performance but require re-training for different spectrometer types and thus depend on reliable annotated datasets. 
Our flexible scoring scheme allows for user controlled training on supplied annotated datasets.
The topology of the network can either be defined by the user or, following the approach proposed by Bern\cite{Bern_Spec_Fusion_2008}, learned from the given dataset directly.
We extend this approach by considering ion intensities and cleavage positions similar to the PepNovo scoring in order to account for shifts of the fragmentation patterns between different $m/z$ regions along the spectrum.

Our software \antilope (\textsc{anti}symmetric path search with \textsc{l}agrangian \textsc{o}ptimization for \textsc{pe}ptide identification), an implementation of the improved approach, is freely available as part of upcoming releases of the open source proteomics library OpenMS\cite{OpenMS}.

The structure of the remainder of this paper is as follows. Section~\ref{sec:method} describes our new method. 
In Section~\ref{sec:results} we compare our tool with the state-of-the-art tools PepNovo, NovoHMM, LutefiskXP and PILOT\@. Finally, in Section~\ref{sec:conclusions}, we discuss our results and future work.

\section{Novel De Novo Peptide Sequencing Algorithm}\label{sec:method}

This section describes our new approach to the \textit{de novo} sequencing problem. 
At first we formally introduce the graph-theoretic formulation and the resulting ILP formulation our method \antilope is based on. Then we present our new algorithmic approach to find the $k$ best solutions of the ILP\@. Finally, we explain the scoring model of \antilope.

\subsection{Graph-Theoretical Formulation}
Bartels introduced the
transformation of a tandem mass spectrum into the so-called
\emph{spectrum graph}, a now commonly used data structure in
graph-theoretical approaches to the \textit{de novo} sequencing
problem\cite{dacvp-denovo:1999,Pepnovo_05,LiuCaiTreeDecomp}, see also
Fig.~\ref{fig:ncgraph}.  Using this data structure, the
original problem amounts to finding a longest path with certain properties in this graph.

When a peptide $P$ is fragmented by collision induced dissociation (CID) it usually breaks along the backbone between two neighboring amino acids into a pair of N-terminal (prefix) and C-terminal (suffix) fragments.
We define the residual mass of $P$ as the sum of the monoisotopic masses of all amino acid residues in $P$.
By parent mass $M_P$ we denote the total mass of $P$, which is the residual mass, plus $18\unit{Da}$ for an additional water molecule.
Depending on the exact fragmentation position, different types of fragment ions are produced that have a certain mass offset compared to the prefix residue mass (PRM) or suffix residue mass. 
Besides the types presented in Fig.~\ref{fig:fragmentation}, also neutral loss variants, e.g., loss of water or ammonia, of several ion types are observed frequently as well as multiply charged ions.
The fragmentation process is still not fully understood and which types are generated with which intensity depends on many factors. 
\begin{figure}[b]
\centering
\includegraphics[width=1.0\linewidth]{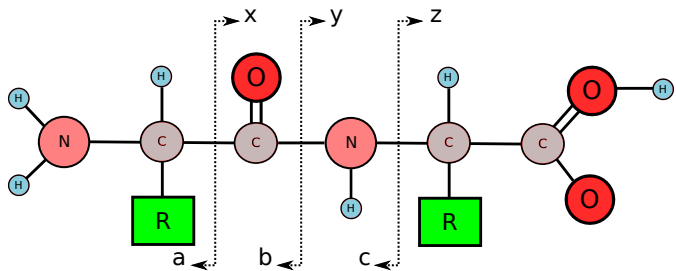}
\caption{Peptide fragmentation along the backbone. 
This figure displays the most prominent fragmentation positions for the generation of pairs of b/y-ion, a/x-ion and c/z-ion in the backbone of a peptide.}
\label{fig:fragmentation}
\end{figure}



The spectrum graph $G$, consists of a set of nodes $V$, a set of directed edges $E_D$ and a set of undirected edges $E_U$.
In the original definition the spectrum graph does not contain the set of undirected edges, which is a slight modification by Liu and Cai\cite{LiuCaiTreeDecomp} who termed this the \emph{extended spectrum graph}.
In the spectrum graph each node corresponds to some possible prefix residue mass of the peptide to be identified.
Directed edges represent amino acids and connect nodes if their mass difference can be explained by some amino acid.
Two nodes that lead to contradicting interpretations of some mass peak are called complementary and are connected by an undirected edge.

Given the tandem mass spectrum of some unknown peptide the construction of the spectrum graph is as follows: Each peak $s$ with mass $m_s$ in the input spectrum generates a set of nodes. If we consider $k$ different N-terminal ion types (e.g., b-ion and a-ion) with mass offsets $\delta_1, \ldots, \delta_k$ ($+1 \unit{Da}$ for b-ions, $-27 \unit{Da}$ for a-ions) from the PRM, then peak $s$ generates $k$ nodes with masses $m_s-\delta_1, \ldots, m_s-\delta_k$.
For C-terminal ion-types with offsets $\delta_1, \ldots, \delta_k$, additional $k$ nodes with masses $M_p-18-(m_s-\delta_1), \ldots, M_p-18-(m_s-\delta_k)$ are generated.
Each of these nodes represents the prefix residue mass under the assumption that $s$ was generated by an ion of a certain type.
Clearly at most one of these nodes can represent the true PRM, therefore they are all contradicting each other and are connected by undirected edges.
Whenever the mass difference of two nodes $v_i$ and $v_k$ equals the mass of some amino acid $\alpha$ ($\pm \epsilon$), we connect $v_i$ and $v_k$ via a directed edge $(v_i,v_k)$ labeled with $\alpha$. Finally we add two so called \emph{goalpost} nodes $s$ and $t$, with masses $0$ and  $P_M-18\unit{Da}$, respectively. 

If the spectrum of some peptide $P$ is complete, i.e., fragment ion peaks are abundant for each possible cleavage site of $P$, then there exists a node for each PRM of $P$.
Therefore the correct sequence of $P$ is obtained by finding the $s$-$t$-path of nodes corresponding to the true prefix sequences of $P$ and by concatenation of the edge labels along this path.
Each node in the spectrum graph has a score that represents the reliability of that node to correspond to a true PRM.

However, simply looking for the longest path in the graph often leads to infeasible solutions, namely if two nodes that were generated by the same peak are included in the path, since in general only one of them corresponds to a true PRM. 
This problem is aggravated when the score of each node is directly related to the intensity of the generating peak.
In such a scenario a high intensity peak generates several high scoring nodes and a longest path search then tends to include a pair of complementary nodes in the longest path leading to a contradicting N- and C-terminal interpretation of the same peak.
Such an infeasible path is called symmetric because the pairs of forbidden pairs of N-terminal and C-terminal nodes form a symmetric structure, which can be seen in Fig.~\ref{fig:ncgraph}.
To solve the \textit{de novo} sequencing problem we hence have to search for \emph{antisymmetric} paths.
These are paths without contradicting nodes. They therefore do not contain pairs of nodes that are connected by an undirected edge. See Fig.~\ref{fig:infeasible_example} for an example.
\begin{figure}[hbt]
\begin{center}
\includegraphics[width=1.0\linewidth]{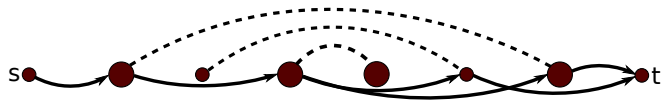}
\caption{Symmetric path example. 
This figure sketches schematically the situation when an infeasible symmetric path would be preferred over a feasible antisymmetric solution.
Assuming that the small nodes have a score of 1 and the bold nodes have a score of 2, the illegal $s$-$t$ path scores higher than the legal one.
Therefore, in this example, a simple longest path search yields infeasible solutions.
}
\label{fig:infeasible_example}
\end{center}
\end{figure}

Most \textit{de novo} sequencing algorithms generate one pair of complementary nodes for each peak assuming it being either a b-ion or y-ion.
These pairs form a nested non-interleaving structure allowing for efficient computation.
But although b- and y-ions are usually the most abundant in CID spectra, there are cases in which both of them are missing and therefore no correct node is generated in this case.
Therefore it is promising to include nodes for other interpretations, especially in the low and high mass range of the spectrum where fragmentation is usually less complete.

While the longest antisymmetric path problem is NP-complete for general directed graphs\cite{Gabow76}, there exist polynomial algorithms for the special case where the non-interleaving property is satisfied.
The polynomial algorithm proposed by Chen\cite{Chen_2001} uses dynamic programming to compute an optimal solution to the longest antisymmetric path with non-interleaving forbidden pairs.
In a second paper Lu and Chen\cite{Chen_subopt_03} extended this approach to compute suboptimal solutions by constructing a so called matrix spectrum graph and applying depth-first search and a backtracking algorithm.
In contrast, the ILP formulation presented in the next section does not depend on such a nested structure and corresponds to the \textit{de novo} sequencing problem for any desired set of ion types.

\subsection{Integer Linear Programming Formulation}\label{sec:complete_ilp}
Our algorithm is based on the following integer linear programming (ILP) formulation\cite{Andreotti2008}, which is very similar to the one Floudas and DiMaggio used for their tool PILOT\cite{Floudas_DiMaggio_07}.
Our formulation models the problem by means of zero-one variables for each edge. We put the score of each node on all its outgoing directed edges. As the graph is acyclic, this is a safe transformation.

\begin{align}
\max\sum_{(v_i,v_k) \in{E_D}} \hspace*{-.5em}c_{i,k} x_{i,k}    \label{form_obj_func}\\
\sum_{(v_s,v_k)\in{E_D}}\hspace*{-.5em}  x_{s,k}&= 1  \label{form_const_src_out}\\
\sum_{(v_k,v_t)\in{E_D}}\hspace*{-.5em} x_{k,t} &= 1  \label{form_const_sink_in}\\
\sum_{(v_i,v_k)\in{E_D}}\hspace*{-.5em} x_{i,k} - \hspace*{-.75em}\sum_{(v_k,v_j)\in{E_D}}\hspace*{-.5em} x_{k,j} &= 0 \qquad \forall k\in{V} \setminus{\{v_s,v_t\}} \label{form_const_flow_cons}\\
\sum_{v_i \in e} \sum_{(v_i,v_k)\in E_D}\hspace*{-.5em} x_{i,k} &\le 1 \qquad \forall e \in{E_U} \label{form_const_antisym}\\
x_{i,k}&\in \{0,1\} \label{form_const_binary}
\end{align}
We introduce a binary variable $x_{i,k}$ for every directed edge $(v_i,v_k)\in E_D$ which has value one if edge $(v_i,v_k)$ is part of the path (active) and zero otherwise (inactive).
The objective function~(\ref{form_obj_func}) maximizes the summed score of all active directed edges.
For the two goalposts $s$ and $t$, the two constraints~(\ref{form_const_src_out}) and~(\ref{form_const_sink_in}) assure that exactly one active edge leaves $s$ and one enters $t$.
Together with the flow conservation constraints~(\ref{form_const_flow_cons}), they establish a correspondence between feasible solutions of the ILP and $s$-$t$ paths in the graph. 
An optimal solution of the ILP consisting of objective function~(\ref{form_obj_func}) and constraints~(\ref{form_const_src_out}), (\ref{form_const_sink_in}) and (\ref{form_const_flow_cons}) corresponds to a longest $s$-$t$ path, still possibly symmetric and therefore infeasible for the \textit{de novo} sequencing problem.
Therefore we add another constraint~(\ref{form_const_antisym}) that makes sure that for each pair of contradicting nodes at most one will be selected. 
The difference of our model to the one proposed by Floudas and DiMaggio is twofold.
First, we do not introduce variables for nodes as they are not required.
This does not change the general structure of the formulation and has no strong effect on the time required for solving.
Second, we do not formulate a constraint that prevents the exact mass of the predicted sequence to deviate from the measured parent mass by more than a certain threshold value (usually $2.5 \unit{Da}$).
We argue that in our algorithm it is more promising to defer this filtering to a later stage of the algorithm.
Since we add edges that correspond to pairs and triples of amino acids which often represent several possible combinations of amino acids, there is no exact mass which could be used in such a constraint.
Therefore we perform the filtering at a later stage when we have created the candidate superset.

\subsection{Applying Lagrangian Relaxation} \label{sec:lagrange} While
linear programming (LP) problems can be solved in polynomial worst
case time, adding integrality constraints makes them generally NP-hard
and the resulting integer linear programs (ILPs) require different
algorithmic solution approaches. One common method is to first solve
the LP relaxation and then investigate the obtained solution. If the
solution is fractional one has to resort to techniques like
branch-and-bound or branch-and-cut using upper and lower
bounds obtained from heuristics and from the relaxed solution. 

We apply a different kind of relaxation method, \emph{Lagrangian relaxation}, which yields in many cases much more efficient algorithms than those based on LP relaxations because it can exploit structural knowledge of the problem.
Lagrangian relaxation is motivated by the experience that many hard integer programming problems correspond to a significantly easier problem that has been complicated by an additional set of constraints.
To obtain the efficiently computable Lagrangian problem, the complicating constraints are removed and replaced by a penalty term in the objective function. The relaxed problem obtained that way is called the Lagrangian problem and can often be solved efficiently.

The Lagrangian relaxation of the \textit{de novo} sequencing ILP~(\ref{form_obj_func})-(\ref{form_const_binary}) is straightforward as it is very obvious that the antisymmetry constraints form the class of \textit{hard} constraints that complicate the computationally \textit{easy} problem of a longest path search in a directed acyclic graph (DAG).
We can solve this relaxed problem by means of a simple standard algorithm, which can be found in reference material~\cite{cormen2001algorithms}.
To make the Lagrangian relaxation more transparent we rewrite the objective function in a way that the edge variables are grouped by the undirected edges incident to their left end:
\[\max \sum_{e \in E_U}{\sum_{\substack{(v_i,v_k) \in{E_D}, \\ v_i \in e}} c_{i,k} x_{i,k}} \enspace. \]

Next we apply Lagrangian relaxation by dropping the antisymmetry constraint~(\ref{form_const_antisym}) and moving it to the objective function to penalize its violation.
This leads to the Lagrangian problem 
\begin{align} \label{form_lagr_prob}
 Z(\lambda) = \max \sum_{e \in E_U}\sum_{\substack{(v_i,v_k)\in{E_D}, \\ v_i \in e}}&\hspace*{-.5em}(c_{i,k} - \lambda_e)x_{i,k} \ + \sum_{e \in E_U} \lambda_e\\
\sum_{(v_s,v_k)\in{E_D}} x_{s,k} &= 1 \nonumber  \\
\sum_{(v_k,v_t)\in{E_D}} x_{k,t} &= 1  \nonumber \\
\sum_{(v_i,v_k)\in{E_D}} x_{i,k} - \hspace*{-.5em}\sum_{(v_k,v_j)\in{E_D}}\hspace*{-.5em} x_{k,j} &= 0 \qquad \forall k\in{V} \setminus{\{v_s,v_t\}} \nonumber\\
x_{i,k}&\in \{0,1\} \nonumber
\end{align}


The vector $\lambda$ holds the Lagrangian multipliers, non-negative real numbers that define the weight of the penalty term.
\begin{lemma}
The Lagrangian problem~(\ref{form_lagr_prob}) can be solved in linear time and space.
\end{lemma}
\begin{IEEEproof}
Solving the Lagrangian problem consist of the following steps:
First we simply subtract from each edge weight $c_{i,k}$ the value $\lambda_e$, for all undirected edges $e$ incident to node $v_i$.
Then we apply the linear time $O(|V|+|E_D|)$ longest path search algorithm for DAGs on the graph with the modified edge weights.
Finally we add the value of $\sum_{e \in E_U} \lambda_e$ to the score obtained from the longest path search algorithm.
Obviously each of the steps requires only linear time and space. 
\end{IEEEproof}
By restricting the Lagrangian multipliers to non-negative values one can easily show that the value of the solution of the Lagrangian problem is an upper bound to the optimal value of the original problem~\cite{Linear_programming:97}.
In order to obtain a tight bound, the strategy is to find the values for the Lagrangian multipliers that minimize $Z(\lambda)$, which means solving the dual problem:
\[Z_D = \min_{\lambda \geq 0} Z(\lambda)\enspace. \] 
We apply the efficient iterative subgradient optimization algorithm, computing sequences of multipliers $\lambda^t$ where $t = 0, 1, 2, \ldots$ denotes the iteration. 
We start with $\lambda^0_e=0$, for all $e \in E_U$ and in each iteration $t$ we compute the subgradients $S^t_e = 1- \sum_{v_i \in e} \sum_{(v_i,v_k)\in E_D} x_{i,k}$, for all $e \in{E_U}$ and update the Lagrangian multipliers according to formula:
\begin{equation}\label{form_lagr_multi}
\lambda^{t+1}_e =  \max \{0, \lambda^t_e -  \theta^t S^t_e\} \enspace.
\end{equation}
One crucial factor with a huge influence on the performance is the step-size $\theta$.
The subgradient method converges to the optimal solution $Z_D$ if the step-size satisfies the following conditions\cite{HeldWolfeCorwder:1974}:
\[ \lim_{k \to \infty} \theta^k = 0 \quad \text{and} \quad \lim_{k \to \infty} \sum_{i=1}^{k}{\theta^i} = \infty \enspace. \]
A formula that is widely used for step-size computation because it shows good performance in practice is given by
\begin{equation}\label{form_lagr_step-size}
\theta^t = \frac{\gamma^t (Z(\lambda^t) - Z^*)} {\sum_{e \in{E_U}}{(S^t_e)^2}} \enspace,
\end{equation}
where $Z^*$ is the value of the best solution to the original problem that was computed yet and $\gamma^t$ defines a decreasing adaption parameter.

\subsection{Suboptimal Solutions}
A straightforward strategy to compute suboptimal solutions, also implemented in PILOT\cite{Floudas_DiMaggio_07}, is to cut off previous solutions by an additional constraint.
A known drawback of this approach is that solving time may increase dramatically after generating a few suboptimal solutions.
We suggest a different strategy and overcome this problem by means of an algorithm which, for a given number $k$, directly computes the $k$ longest paths.
We use the algorithm by Yen\cite{yen_k_shortest} that was originally designed to compute the $k$ shortest paths without cycles on general directed graphs.

Yen's algorithm is a deviation algorithm based on the fact that the $i$-th shortest path $P_i$, will coincide with every shorter path $P_{i-1} \ldots P_1$ up to some node until it deviates.
The farthest node from the source $s$ with this property is called \emph{deviation node} $d(P_i)$.

The strategy to find the $i+1$-st shortest $s$-$t$ path $P_{i+1}$ is, starting at $d(P_i)$, to compute for each node $v^i_j$ of $P_i$ the shortest path to $t$, that deviates from $P_i$ at node $v^{i}_j$.
Therefore a shortest path from $v^{i}_j$ to $t$ is computed which is not allowed to use the edge ($v^{i}_j$, $v^{i}_{j+1}$).
This shortest path from $v^{i}_j$ to $t$ is then concatenated with the prefix $(v^i_1\ldots v^i_{j-1})$ of $P_i$ to obtain the shortest $s$-$t$ path that deviates from $P_i$ at node $v^{i}_j$.
This path is added to a candidate set $X$.
After the shortest deviating paths of $P_i$ have been computed, the shortest path in the candidate set $X$ corresponds to the $i+1$-st shortest $s$-$t$ path $P_{i+1}$ and is removed from $X$.

Yen's algorithm performs an additional trick to guarantee for paths without cycles that we do not discuss here.
For a more detailed description of this algorithm and variants please refer to reference material\cite{yen_k_shortest,Martins2003}.

Our problem differs in a few points from the original problem solved by Yen's algorithm, so it requires a few adaptations.
While Yen's algorithm is designed for general directed graphs that may contain cycles, we are working on a DAG\@.
This simplifies the problem as we do not have to worry about cycles and can simply transform the shortest path problem into a longest path problem. Note that the longest path problem is NP-complete in graphs with cycles. A second difference is that we have the additional condition to find antisymmetric paths.
Therefore every time the shortest path algorithm is called in the Yen's algorithm, we replace this by solving the Lagrangian relaxation formulation for the longest antisymmetric path search.
The following theorem and its proof capture the main algorithmic result of this paper.

\begin{theorem}
The combination of our Lagrangian relaxation-based algorithm for antisymmetric paths and a modification of Yen's algorithm solves the problem of computing the $k$ longest antisymmetric paths in time $O(k l s(|E|+|V|))$, where $l$ is the length of the longest path and $s$ is the total number of subgradient iterations.
\end{theorem}

\begin{proof}
In iteration $i+1$ of Yen's algorithm the computed path deviating from $P_i$ at node $v^{i}_j$ must satisfy two conditions in order to form an antisymmetric path in $G$.
\begin{enumerate}
\item There are no two nodes in the path from $v^{i}_j$ to $t$ that are in conflict.
\item None of the nodes in the computed path from $v^{i}_j$ to $t$ is in conflict with some node from the prefix of path $P_{i}$ up to node $v^{i}_j$.
\end{enumerate}
The first condition is satisfied by the Lagrangian relaxation formulation itself, because if applied to the subgraph $v^{i}_j \ldots t$, every feasible solution corresponds to an antisymmetric path from node $v^{i}_j$ to $t$.
To meet the second condition it is sufficient to remove all nodes from the subgraph $v^{i}_j \ldots t$ that are connected via an undirected edge with some node of the prefix $s\ldots v^{i}_j$ of $P_{i}$ before we compute the longest antisymmetric path.
This trick also simplifies the longest antisymmetric path search for increasing $j$ as the possibilities to generate an infeasible solution are decreasing.

The complexity of Yen's algorithm for computing the $k$-longest paths in a DAG is $O(k |V| (|E|+|V|))$.
The first factor $|V|$ comes from the fact that, in a general graph, one path can possibly contain all $|V|$ nodes.
In the case of peptide sequencing, the length of a path equals the length of the predicted peptide which usually does not exceed a length of $30$ for typical experimental settings.
In the longest antisymmetric path version using Lagrangian relaxation, the $O(|E|+|V|)$ DAG longest path algorithm gets iteratively called during subgradient optimization algorithm.
Therefore the complexity of our formulation for identification of a peptide containing $l$ amino acids is $O(k l s(|E|+|V|))$ with $s$ being the number of iterations during subgradient optimization.
\end{proof}
Note that the value of $s$ is possibly exponential if the subgradient optimization does not converge and the complete branch and bound tree has to be enumerated.
Nevertheless, in the results section we will show that for our peptide sequencing formulation on average only very few iterations are required which leads to a practically efficient algorithm.

\subsection{Scoring Model}
\label{sec:scoring_model}
We use a probabilistic scoring based on a Bayesian network similar to the scoring model of PepNovo.
Bayesian networks are directed acyclic graphs where nodes represent random variables and the edges represent conditional dependencies between variables.
The variables in our model are the ion types $t\in T$ that are considered by our scoring model and the possible values for each variable is the intensity.
Therefore, as a first step, we normalize the intensity of all peaks to discrete values as defined by Dan\v{c}{\'i}k et al.\cite{dacvp-denovo:1999} by using their rank as intensity.

The usage of Bayesian networks for scoring nodes in the spectrum graph is motivated by the observation that fragmentation events are not independent.
For example, the probability of observing a strong b-ion is not independent of the abundance and intensity of the complementary y-ion.

Unlike for the PepNovo algorithm, where the structure of the probabilistic network is predefined leading to a fixed set of accounted conditional dependencies, we implemented a flexible scoring scheme where the network topology can be either defined by the user or it can be learned during the training process automatically.

For inference and training of the Bayesian network we used the Bayesian Network Classifiers in the machine learning suite Weka\cite{WEKA:2009}.
Similar to PepNovo, we discretize the relative position of a cleavage into several (default 3) equally sized regions $r$ to account for the different intensity distributions in the center and terminal regions usually observed in CID spectra.
For each of the regions we train a Bayesian network using some training set of tandem mass spectra with known peptide identification.
For each training spectrum we construct the node set of the spectrum graph and select an equal number of true positives (vertices representing a true PRM) and false positives (vertices representing not representing PRM). 
For each of the selected nodes we look for witnessing peaks at their calculated positions and record their normalized intensity to obtain the training vectors for the Bayesian network.
Each of the training vectors has one additional entry, the class label, which is $T$ for true positives and $F$ for false positives. 
We select only those ion types of the witness set for the network training that appear in at least $t$ percent of the true positive samples of the training set where the threshold parameter $t$ can be defined by the user.
For each of these selected types we then add a node in the Bayesian network. 
One additional node for the class is added.
During the network training, the structure (set of directed edges) of the network is learned and once the structure is fixed the conditional probability tables are learned from the training data.
While the user can control a huge range of possible options for the Weka Bayesian network classifier training through our program, we set as the default training algorithm the K2-HillClimber and the Bayesian metric for local scoring\cite{WEKA:bayes}.
For a user defined network topology the first step is skipped and only the conditional probability tables are computed.

Once the network is trained, we score a node $v$ in the spectrum graph by looking for peaks in the spectrum at the calculated masses for the selected ion types to obtain the set of intensity observations $I^v$. 
Using the trained Bayesian network \textit{BN} we then compute the log likelihood ratio as:
\begin{equation}
\textit{LLR}(v)= \log \frac{\Pr(I^v \mid \textit{BN}, \text{class}=\textit{T}) } {\Pr(I^v \mid \textit{BN}, \text{class}=\textit{F})}\enspace, 
\end{equation}
where $\Pr(I^v \mid \textit{BN}, \text{class}=\textit{X} \in\{T,F\})$ is the probability of observation $I^v$ under model \textit{BN} when the class variable is set to $X$.
In contrast to Bern and Datta who obtain their false positive samples from perturbation of the correct PRM we only take false positive PRM for which a node was created during the spectrum graph construction. 
We chose that approach since we want the Bayesian network to discriminate between correct and false nodes in the spectrum graph.
By just perturbing the true PRM one will very likely generate false positive training samples containing only zero intensity entries which will never be the case for a node in the spectrum graph as it requires at least one peak to be generated.

Additional to the Bayesian network we also use a simple intensity rank score $S_R(v)$ as it is also used by INSPECT\cite{Tanner_Inspect_2005}.
This score is the ratio between two probabilities, the probability that a peak with a certain intensity rank corresponds to a certain ion type (e.g., a b-ion) and the the probability that a randomly chosen peak was generated by that ion type.
As these values differ between different mass regions of a spectrum, we split the spectrum into three equally spaced mass regions and estimate the probabilities for each of them separately using the same training data as for the Bayesian network.
For example, if we generate a node for a peak of rank 4, and this node interprets the peak as a b-ion, then $S_R(v)$ is the log ratio between the probability that a rank 4 peak is a b-ion and the probability that any random peak is a b-ion. 

The final score $s(v)$ for each node $v$ of the spectrum graph is then computed as:
\begin{equation}
s(v)= \textit{LLR(v)} + S_R(v)
\end{equation}

Nodes having negative scores correspond to unreliable PRMs and are removed from the graph in order to reduce the size of the spectrum graph and speed up the candidate generation process.
Since our formulation is working with edge weights, we move the node scores onto the edges, such that each directed edge gets the score of its left node.
In the filtered spectrum graph we compute the predefined number of suboptimal solutions, each corresponding to one antisymmetric path.
To account for missed cleavages we also add edges corresponding to pairs and triples of amino acids to the spectrum graph.
For each of the generated candidates, in a second step, we try to resolve the pairs and triples of amino acids.
Therefore we generate all possible combinations and permutations of amino acids to generate a candidate superset.
The candidates in that superset are then re-scored by a refined version of a shared peaks count where we reward abundant witness peaks and penalize missing ones.

Given a candidate sequence we look for witnessing peaks in the query spectrum and give a bonus if one was found or a penalty if it is missing. 
Further we check whether the peak is a primary isotopic peak, a secondary isotopic peak or a lone peak. 
A peak is called a primary isotopic peak if we find a child peak at offset 1 Da for a singly charged ion or 0.5 Da for a doubly charged ion.
Equivalently a peak is called a secondary isotopic peak if it has a parent peak with offset -1 Da for a singly charged ion or -0.5 Da for a doubly charged ion.
If a peak is neither primary isotopic nor secondary isotopic then it is a lone peak.
If a witnessing peak is a primary peak we add another bonus to the score while we charge a penalty if it is a secondary peak.   
While the reward and penalty score are actually user parameters we will offer a generic algorithm to estimate reasonable values in future versions.
The candidates are then re-ranked according to this score and the predefined number of candidates is returned.

\section{Results}

In this section we present and discuss our computational results. We compared \antilope with state-of-the-art alternative peptide identification software with respect to running time and quality.

\label{sec:results}
\subsection{Efficiency}
The major contribution of this work is the new algorithmic approach based on Lagrangian relaxation.
We will first give a thorough analysis of the performance and compare it to the ILP formulation (1)-(6). 
Like implemented in PILOT, we generate the suboptimal solutions by introducing additional constraints that cut off previous solutions.
We implemented our algorithm in C++ and use the OpenMS\cite{OpenMS} library that offers convenient data structures and algorithms to handle and manipulate spectral data.
For the ILP formulation we use the commercial CPLEX\cite{cplex} solver software (version 9.0), which is in general the fastest solver available. We were not able to directly compare to PILOT because the software is not available upon request.  
In Fig.~\ref{fig:runtime_comp} we compare the running times of the ILP and our Lagrangian relaxation formulation on a set of 100 tandem mass spectra from the ISB dataset\cite{Keller2002}.
In this comparison we only consider the time required to generate the set of candidate sequences and ignore the preprocessing of the spectrum and the spectrum graph generation as these steps are independent of the applied algorithm.
We compared the running time required to generate the top scoring 20, 30 and 50 candidates for each spectrum.
The figure shows that our approach significantly outperforms the ILP formulation on all instances and the performance gain increases with the number of candidates to be generated.
Our algorithm is on average $\approx 9$ times faster for the best 20 candidates, for 30 and 50 candidates the average advantage increases to a factor of $\approx 12$ and $\approx 18$.
While the run time for of the ILP formulation for the top 50 candidates was usually above 2 seconds, the Lagrangian relaxation formulation requires on average only a few tenths of a second.

\begin{figure}[hbt] 
\centering
\includegraphics[width=1.0\linewidth]{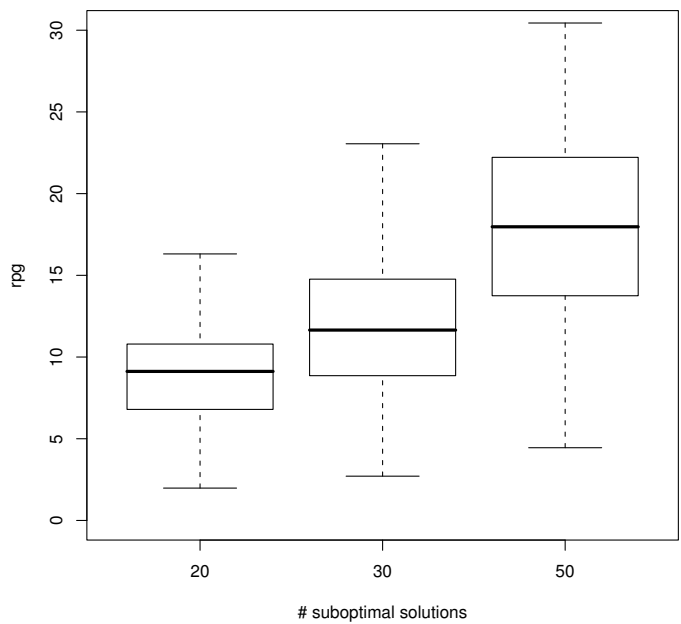}
\caption{Running time comparison between Lagrangian relaxation and ILP formulation for computation of 20, 30, and 50 suboptimal solutions of 100 benchmark spectra. Box-and-whisker plots display median, quartiles, and extrema of the distribution of relative performance gains $\textit{rpg} = \text{run-time(ILP)}/\text{run-time(Lagrange)}$.
\antilope outperforms the CPLEX-based method for all spectra and all numbers of suboptimal solutions. The advantage increases with
the number of suboptimal solutions. The considered spectrum graphs contained between 80 and 200 nodes.}
\label{fig:runtime_comp}
\end{figure}

In a closer analysis we investigated the convergence behavior of our Lagrangian relaxation formulation. It reveals that for each Lagrange problem solved during the path ranking algorithm only very few iterations of the subgradient optimization are required.
The path ranking algorithm maintains a list of previously detected candidate paths together with their scores.
Since the score $Z(\lambda)$ of the Lagrange problem is an upper bound to the score $Z_{\text{IP}}$ of the best possible feasible solution, subgradient optimization can be aborted as soon as $Z(\lambda)$ falls below the lowest score in the candidate list.
If the Lagrangian relaxation does not converge and cannot be aborted after 100 iterations we apply a branching step.
We use the best infeasible path found during the subgradient optimization and arbitrarily choose one node $v_b$ involved in a conflict.
Then we generate two subproblems, one forcing $v_b$ to be in the path and one forbidding $v_b$ to be selected.
We found that only for a very small fraction of the Lagrange problems a branching step had to be performed, and the depth of branch-and-bound trees never exceeded a value of three.
It is necessary to mention that the performance strongly depends on the scoring function used, since a good scoring function will not generate many high scoring nodes for the same peak and only the correct one should receive a significantly high score. Therefore a good scoring function does not only affect the identification performance but also affects the complexity of the candidate generation.

\subsection{Sequencing Performance}
\label{sec:sequencing_performance}

\label{sec:evaluation}
\begin{figure*}[h!t] 
\centering
\includegraphics[width=1.0\linewidth]{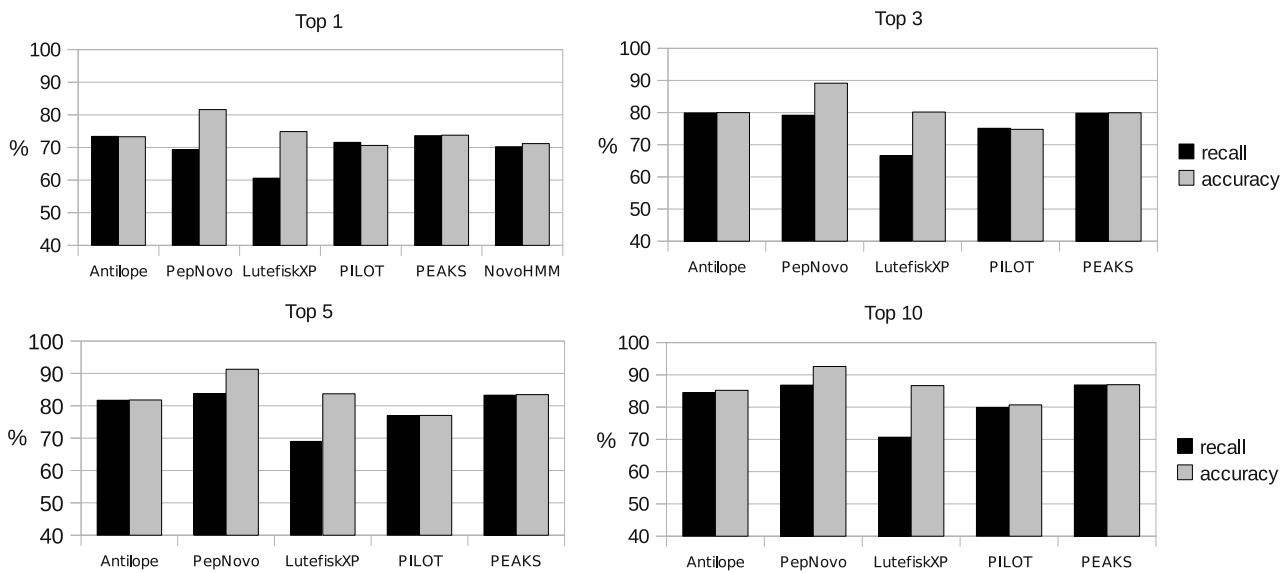}
\caption{Benchmark. Comparison of accuracy and recall of \antilope with NovoHMM, PepNovo, PILOT, PEAKS and LutefiskXP\@. We compare the accuracy and recall of the best prediction among the 
top 1, 3, 5 and 10 ranked candidates returned by each tool. 
As the best prediction we consider the one with the best recall among the candidates.
Since NovoHMM generates only one candidate per spectrum it appears only in the first plot.
Discussion in text.
}
\label{fig:eval_stats}
\end{figure*}

We compared the performance of \antilope to four non-commercial \textit{de novo} sequencing tools, LutefiskXP, NovoHMM, PILOT\footnote{As PILOT was not available, the identifications for the test data were generated by the authors of PILOT.}, PepNovo and the commercial software PEAKS\footnote{We used the PEAKS Online 2.0 web interface.}.
We used two measures, accuracy and recall,  to assess their performance.
Accuracy denotes the fraction of correctly predicted amino acid residues compared to all predicted residues.
Recall is the fraction of correctly predicted residues compared to the total number of residues of the correct peptide sequences.
When looking at suboptimal solutions for each algorithm we looked for the prediction with the highest recall and reported the values of this prediction for recall and precision.
In case of multiple predictions with the same recall value we report the values for the one with the highest precision among them.
As benchmark set, we chose tandem mass spectra from the ISB dataset~\cite{Keller2002} that were generated by an ESI-ion trap mass spectrometer by Thermo Finnigan and spectra from the open proteomics database.
This set of reliably annotated spectra from tryptic peptides has already been used for training of PepNovo and NovoHMM\@.
We created a training set of 1214 spectra from doubly charged precursor ions of unique peptides to train the scoring model.
During the Bayesian network training for each of the 3 mass sectors, only the ion types that had a peak in at least $20\%$ of the true positive training samples are selected for the corresponding Bayesian network. The topologies of the Bayesian networks together with a brief discussion can be found in the supplementary material.
The parameters to score the peptide spectrum matches for the candidate sequences in the superset were chosen as follows:
For an abundant b- or y-ion we awarded the score $\textit{PSM}_b = \textit{PSM}_y = 1$, doubly charged b- or y-ions scored $0.5$, a-ions $0.3$ and all neutral losses were awarded a score of $0.2$. 
Isotopic peaks for some type $t$ were awarded a score of $\textit{PSM}_t \cdot 0.2$. 
If some peak was missing the penalty of $\textit{PSM}_t \cdot 0.5$ was subtracted from the score.
When a peak was classified as a secondary peak its score is reduced to $\textit{PSM}_t \cdot 0.8$.
The score for some peak is then weighted with the relative $m/z$ distance between the expected and the observed $m/z$ value using a linear function.

The test set consists of 200 spectra of peptides (peptides in training and test dataset disjunct) with a molecular mass of at most 1600\unit{Da} and an average peptide length of 10 residues.
We score a predicted amino acid as correct, if its predicted starting mass position does not deviate by more than $2.5\unit{Da}$ from the correct starting mass position.
Further, in our evaluation we do not discriminate between the amino acids Q/K and and I/L since their masses cannot be distinguished.

To compare the tools we do not only look at the top hit, but we also look at the accuracy and recall for the best hit in the top 3, 5 and 10 candidates.
The results are presented in Fig.~\ref{fig:eval_stats}.
Since NovoHMM only generates one candidate per spectrum it appears only in the first plot.
Looking only at the top hit, the recall of \antilope ($\approx73.4\%$) is only marginally lower than of PEAKS ($\approx73.7\%$) but slightly better than that of PILOT ($\approx71.5\%$), NovoHMM ($\approx70.2\%$) and PepNovo ($\approx69.4\%$). 
The recall of LutefiskXP ($\approx60.6\%$) is much lower than for all other tools.
Since \antilope and NovoHMM both compute complete sequences, they have almost equal values for accuracy and recall, while for LutefiskXP and PepNovo these values differ as they allow for gaps in their predicted sequences. 
If we go over from the top hit to the best 3, 5 and 10 candidates, we observe that in terms of recall, \antilope is always very close to PepNovo and PEAKS (equal for the top 3, $2.5\%$ advantage of PepNovo and PEAKS for the top 10) and always approximately $4\%$ better than PILOT\@. 
In terms of accuracy PepNovo ($\approx92\%$) has a better performance than \antilope, PEAKS and PILOT since it allows for partial peptide predictions.
The accuracy of LutefiskXP is slightly better than for \antilope and PILOT but this accuracy is achieved at a much lower recall which is between $12\%$ and $14\%$ lower in all four cases.
The four tools \antilope, LutefiskXP, NovoHMM and PepNovo are comparable in terms of run time which is usually between 0.5 and 1.5 seconds per spectrum. 
The running time of PILOT as reported by the authors was on average around 9 seconds per spectrum.
We cannot directly estimate the runtime of PEAKS since the identification is performed via a web interface.


\section{Conclusion}
\label{sec:conclusions}
We proposed a new algorithmic approach to solve the longest antisymmetric path problem by means of Lagrangian relaxation, combined with a polynomial algorithm for suboptimal solutions.
Using this approach the algorithm is flexible and not restricted to the nested structure of the spectrum graph and solves this problem much faster than an LP relaxation-based method for the same formulation.
Therefore, for our tool \antilope, the candidate generation is no longer the bottleneck as the most time consuming step is the re-ranking phase since the number of possible candidates can easily explode if several double and triple amino acid edges are selected.
In terms of sequencing performance, \antilope is already competitive to available state-of-the-art programs PepNovo and PEAKS while it outperforms LutefiskXP and NovoHMM especially if we also consider suboptimal solutions.
For long peptides PepNovo still has a small advantage, which is mostly due to the fact that the current version of \antilope produces only complete annotations without gaps.

Actually we only generated two nodes for each peak, one for a b- and one for a y-ion.
Generating nodes for all ion types decreased the performance as this always lead to some high scoring, but false nodes and thus to wrong interpretations.
Nevertheless we are sure that generating more nodes can lead to better identifications in combination with a refined scoring scheme.
The algorithmic framework is flexible enough to work with mass spectra generated by different kind of mass spectrometers.
So the user can define for which ion types a node shall be generated.
This can lead to improved identification performance for different datasets.
Combined with a scoring function trained on a representative set of spectra, the ability of our algorithm to directly model multiply charged ions can lead to an improvement over the other algorithms when analyzing tandem mass spectra obtained from higher charged precursor ions.

For the future we plan to improve our algorithm in several directions.
We will include support for identification of peptides containing post-translational modifications.
Further we want to support combinations of complementary fragmentation techniques like CID together with electron transfer dissociation (ETD) or CID with electron capture dissociation (ECD), which can improve the identification\cite{Bern_Spec_Fusion_2008, compnovo}.
In these applications the flexibility of our formulation may become a major advantage over the other programs.

To improve the performance for spectra of longer peptides we will extend \antilope in a way that it can produce partial predictions allowing for gaps at the terminals.
This, together with a machine learning strategy for the re-scoring like the rank-boosting algorithm used by PepNovo, should lead to a further improvement.
\antilope is freely available as part upcoming releases of the open source proteomics library OpenMS\cite{OpenMS} allowing for convenient integration into experimental workflows.

\bibliographystyle{IEEEtran}
\bibliography{denovo}

\begin{thebibliography}{10}
\providecommand{\url}[1]{#1}
\csname url@samestyle\endcsname
\providecommand{\newblock}{\relax}
\providecommand{\bibinfo}[2]{#2}
\providecommand{\BIBentrySTDinterwordspacing}{\spaceskip=0pt\relax}
\providecommand{\BIBentryALTinterwordstretchfactor}{4}
\providecommand{\BIBentryALTinterwordspacing}{\spaceskip=\fontdimen2\font plus
\BIBentryALTinterwordstretchfactor\fontdimen3\font minus
  \fontdimen4\font\relax}
\providecommand{\BIBforeignlanguage}[2]{{%
\expandafter\ifx\csname l@#1\endcsname\relax
\typeout{** WARNING: IEEEtran.bst: No hyphenation pattern has been}%
\typeout{** loaded for the language `#1'. Using the pattern for}%
\typeout{** the default language instead.}%
\else
\language=\csname l@#1\endcsname
\fi
#2}}
\providecommand{\BIBdecl}{\relax}
\BIBdecl

\bibitem{Tanner_Inspect_2005}
S.~Tanner, H.~Shu, A.~Frank, L.~C. Wang, E.~Zandi, M.~Mumby, P.~A. Pevzner, and
  V.~Bafna, ``Inspect: identification of post translationally modified peptides
  from tandem mass spectra.'' \emph{Anal. Chem.}, vol.~77, no.~14, pp.
  4626--4639, 2005.

\bibitem{Sequest_1994}
J.~K. Eng, A.~L. McCormack, and J.~R. Yates, ``An approach to correlate tandem
  mass spectral data of peptides with amino acid sequences in a protein
  database,'' \emph{J. Am. Soc. Mass. Spectrom.}, vol.~5, no.~11, pp. 976--989,
  1994.

\bibitem{Mascot_99}
D.~N. Perkins, D.~J. Pappin, D.~M. Creasy, and J.~S. Cottrell,
  ``{P}robability-based protein identification by searching sequence databases
  using mass spectrometry data.'' \emph{Electrophoresis}, vol.~20, no.~18, pp.
  3551--3567, 1999.

\bibitem{OMSSA}
L.~Y. Geer, S.~P. Markey, J.~A. Kowalak, L.~Wagner, M.~Xu, D.~M. Maynard,
  X.~Yang, W.~Shi, and S.~H. Bryant, ``Open mass spectrometry search
  algorithm,'' \emph{J. Proteome Res.}, vol.~3, no.~5, pp. 958--964, 2004.

\bibitem{PEAKS_03}
B.~Ma, K.~Zhang, C.~Hendrie, C.~Liang, M.~Li, A.~Doherty-Kirby, and G.~Lajoie,
  ``Peaks: Powerful software for peptide de novo sequencing by {MS/MS},''
  \emph{Rapid Commun. Mass Spectrom}, vol.~17, pp. 2337--2342, 2003.

\bibitem{Pepnovo_05}
A.~Frank and P.~Pevzner, ``{PepNovo}: De novo peptide sequencing via
  probabilistic network modeling,'' \emph{Anal. Chem.}, vol.~77, no.~4, pp.
  964--973, 2005.

\bibitem{NovoHMM_05}
B.~Fischer, V.~Roth, F.~Roos, J.~Grossmann, S.~Baginsky, P.~Widmayer,
  W.~Gruissem, and J.~M. Buhmann, ``{NovoHMM}: A hidden markov model for de
  novo peptide sequencing,'' \emph{Anal. Chem.}, vol.~77, no.~22, pp.
  7265--7273, 2005.

\bibitem{Lutefisk_97}
J.~A. Taylor and R.~S. Johnson, ``Sequence database searches via de novo
  peptide sequencing by tandem mass spectrometry.'' \emph{Rapid Commun. Mass
  Spectrom.}, vol.~11, no.~9, pp. 1067--1075, 1997.

\bibitem{dacvp-denovo:1999}
V.~Dan\v{c}{\'\i}k, T.~A. Addona, K.~R. Clauser, J.~E. Vath, and P.~Pevzner,
  ``De novo protein sequencing via tandem mass-spectrometry,'' \emph{J. Comput.
  Biol.}, vol.~6, pp. 327--341, 1999.

\bibitem{Bern_EigenMS_2006}
M.~Bern and D.~Goldberg, ``De novo analysis of peptide tandem mass spectra by
  spectral graph partitioning.'' \emph{J. Comput. Biol.}, vol.~13, no.~2, pp.
  364--378, 2006.

\bibitem{Floudas_DiMaggio_07}
P.~A. DiMaggio and C.~A. Floudas, ``De novo peptide identification via tandem
  mass spectrometry and integer linear optimization,'' \emph{Anal. Chem.},
  vol.~79, no.~4, pp. 1433--1446, 2007.

\bibitem{Bartels:1990}
C.~Bartels, ``Fast algorithm for peptide sequencing by mass spectroscopy,''
  \emph{Biological Mass Spectrometry}, vol.~19, no.~6, pp. 363--368, 1990.

\bibitem{Gabow76}
H.~N. Gabow, S.~N. Maheshwari, and L.~J. Osterweil, ``On two problems in the
  generation of program test paths,'' \emph{IEEE Trans. Softw. Eng.}, vol.~2,
  no.~3, pp. 227--231, 1976.

\bibitem{Chen_2001}
T.~Chen, M.~Y. Kao, M.~Tepel, J.~Rush, and G.~M. Church, ``A dynamic
  programming approach to de novo peptide sequencing via tandem mass
  spectrometry.'' \emph{J. Comput. Biol.}, vol.~8, no.~3, pp. 325--337, 2001.

\bibitem{Chen_subopt_03}
B.~Lu and T.~Chen, ``A suboptimal algorithm for de novo peptide sequencing via
  tandem mass spectrometry.'' \emph{J. Comput. Biol.}, vol.~10, no.~1, pp.
  1--12, 2003.

\bibitem{LiuCaiTreeDecomp}
C.~Liu, Y.~Song, B.~Yan, Y.~Xu, and L.~Cai, ``Fast de novo peptide sequencing
  and spectral alignment via tree decomposition,'' in \emph{Proc.\ 11th Pacific
  Symp Biocomp (PSB)}.\hskip 1em plus 0.5em minus 0.4em\relax World Scientific,
  2006, pp. 255--266.

\bibitem{Bafna_Edwards_03}
V.~Bafna and N.~Edwards, ``On de novo interpretation of tandem mass spectra for
  peptide identification,'' in \emph{Proc.\ 7th Ann Intern Conf Res Comp Mol
  Bio (RECOMB)}.\hskip 1em plus 0.5em minus 0.4em\relax ACM Press, 2003, pp.
  9--18.

\bibitem{DBLP:journals/jco/AlthausC08}
E.~Althaus and S.~Canzar, ``A {Lagrangian} relaxation approach for the multiple
  sequence alignment problem,'' \emph{J. Comb. Optim.}, vol.~16, no.~2, pp.
  127--154, 2008.

\bibitem{Caprara2004a}
A.~Caprara, R.~Carr, S.~Istrail, G.~Lancia, and B.~Walenz, ``1001 optimal {PDB}
  structure alignments: integer programming methods for finding the maximum
  contact map overlap.'' \emph{J. Comput. Biol.}, vol.~11, no.~1, pp. 27--52,
  2004.

\bibitem{Bauer2007}
M.~Bauer, G.~W. Klau, and K.~Reinert, ``Accurate multiple sequence-structure
  alignment of {RNA} sequences using combinatorial optimization.'' \emph{BMC
  Bioinf}, vol.~8, no.~1, p. 271, 2007.

\bibitem{Threading_Andonov:04}
R.~Andonov, S.~Balev, and N.~Yanev, ``Protein threading: From mathematical
  models to parallel implementations,'' \emph{INFORMS J.\ on Computing},
  vol.~16, no.~4, pp. 393--405, 2004.

\bibitem{Bern_Spec_Fusion_2008}
R.~Datta and M.~Bern, ``Spectrum fusion: using multiple mass spectra for de
  novo peptide sequencing.'' \emph{J. Comput. Biol.}, vol.~16, no.~8, pp.
  1169--1182, 2009.

\bibitem{OpenMS}
M.~Sturm, A.~Bertsch, C.~Groepl, A.~Hildebrandt, R.~Hussong, E.~Lange,
  N.~Pfeifer, O.~Schulz-Trieglaff, A.~Zerck, K.~Reinert, and O.~Kohlbacher,
  ``Open{MS} - an open-source software framework for mass spectrometry,''
  \emph{BMC Bioinf.}, vol.~9, p. 163, 2008.

\bibitem{Andreotti2008}
S.~Andreotti, ``Fast de novo sequencing with mathematical programming,''
  Master's thesis, Freie Universit{\"a}t Berlin, Germany, January 2008.

\bibitem{cormen2001algorithms}
T.~H. Cormen, C.~E. Leiserson, R.~L. Rivest, and C.~Stein, \emph{Introduction
  to Algorithms}, 2nd~ed.\hskip 1em plus 0.5em minus 0.4em\relax The MIT Press,
  2001.

\bibitem{Linear_programming:97}
D.~Bertsimas and J.~Tsitsiklis, \emph{Introduction to Linear
  Optimization}.\hskip 1em plus 0.5em minus 0.4em\relax Athena Scientific,
  1997.

\bibitem{HeldWolfeCorwder:1974}
M.~Held, P.~Wolfe, and H.~D. Crowder, ``Validation of subgradient
  optimization,'' \emph{Mathematical Programming}, vol.~6, pp. 62--88, 1974.

\bibitem{yen_k_shortest}
J.~Y. Yen, ``Finding the {$K$} shortest loopless paths in a network,''
  \emph{Management Science}, vol.~17, pp. 712--716, 1971.

\bibitem{Martins2003}
E.~Martins and M.~Pascoal, ``{A new implementation of Yen's ranking loopless
  paths algorithm},'' \emph{Quarterly Journal of the Belgian, French and
  Italian Operations Research Societies}, vol.~1, pp. 121--133, 2003.

\bibitem{WEKA:2009}
\BIBentryALTinterwordspacing
M.~Hall, E.~Frank, G.~Holmes, B.~Pfahringer, P.~Reutemann, and I.~H. Witten,
  ``The {WEKA} data mining software: an update,'' \emph{SIGKDD Explor. Newsl.},
  vol.~11, pp. 10--18, November 2009. [Online]. Available:
  \url{http://doi.acm.org/10.1145/1656274.1656278}
\BIBentrySTDinterwordspacing

\bibitem{WEKA:bayes}
R.~R. Bouckaert, ``Bayesian network classifiers in weka,'' \emph{(Working paper
  series. University of Waikato, Department of Computer Science)}, vol.~14,
  October 2004.

\bibitem{cplex}
{IBM}, ``\textsc{cplex},'' \url{http://www.cplex.com}.

\bibitem{Keller2002}
A.~Keller, S.~Purvine, A.~I. Nesvizhskii, S.~Stolyar, D.~R. Goodlett, and
  E.~Kolker, ``Experimental protein mixture for validating tandem mass spectral
  analysis,'' \emph{OMICS: A Journal of Integrative Biology}, vol.~6, pp.
  207--212, 2002.

\bibitem{compnovo}
A.~Bertsch, A.~Leinenbach, A.~Pervukhin, M.~Lubeck, R.~Hartmer, C.~Baessmann,
  Y.~A. Elnakady, R.~M\"uller, S.~B\"ocker, C.~G. Huber, and O.~Kohlbacher,
  ``De novo peptide sequencing by tandem {MS} using complementary {CID} and
  electron transfer dissociation,'' \emph{Electrophoresis}, vol.~30, no.~21,
  pp. 3736--3747, 2009.

\end{thebibliography}

\begin{IEEEbiography}[{\includegraphics[width=1in,height=1.25in,clip,keepaspectratio]{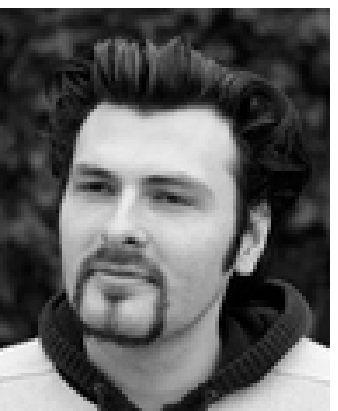}}]{
Sandro Andreotti}
received the MSc degree in Bioinformatics from Freie Universit\"at Berlin, Germany, in 2008 where he is currently working as PhD student in the Algorithmic Bioinformatics group of Knut Reinert. His research is focussing on computational proteomics and discrete optimization.  
\end{IEEEbiography}
\begin{IEEEbiography}[{\includegraphics[width=1in,height=1.25in,clip,keepaspectratio]{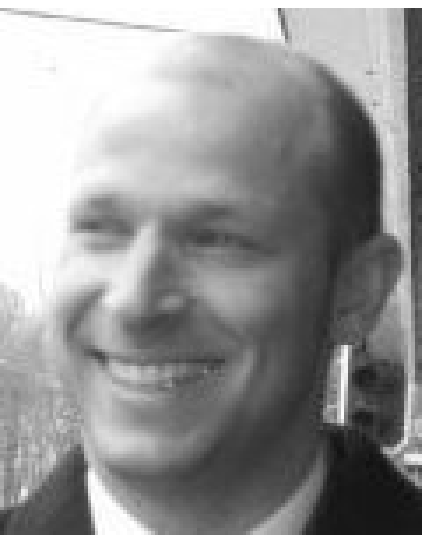}}]{
Gunnar W.\ Klau}
received his PhD in Computer Science in 2001 from Saarland University, Germany. Currently, he heads the Life Sciences group at CWI, the national research center for mathematics and computer science in the Netherlands. His research interests are combinatorial algorithms and discrete optimization in computational biology.
\end{IEEEbiography}
\begin{IEEEbiography}[{\includegraphics[width=1in,height=1.25in,clip,keepaspectratio]{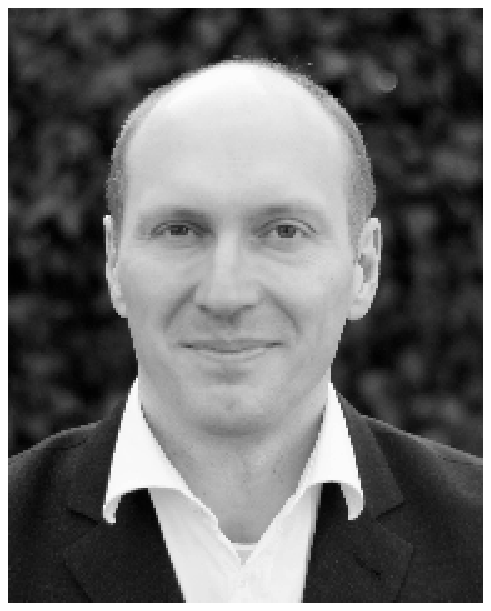}}]{Knut Reinert}
received his PhD 1999 from Saarland University, Germany. He
worked from 1999-2002 for Celera Genomics and took part in the sequencing of the human genome. Currently he is
professor for Algorithmic Bioinformatics at the FU Berlin in Germany. His
research interests lie in developing algorithms for sequence analysis and
proteomics. 
\end{IEEEbiography}
\vfill\end{document}